\documentclass[twocolumn,superscriptaddress,nofootinbib,preprintnumbers,showpacs,tightenlines,notitlepage]{revtex4}
\usepackage[latin1]{inputenc}
\usepackage{graphicx}
\usepackage{amssymb}
\usepackage{color}
\usepackage{float}
\usepackage{amsmath}
\usepackage{amsfonts}
\usepackage{dcolumn}
\usepackage{hyperref}
\usepackage{amsthm}
\usepackage{color}

\def\uudot{\dot{u}}
\def\3nab{\tilde{\nabla}}

\def\la {\langle}
\def\ra {\rangle}

\def\be {\begin{equation}}
\def\ee {\end{equation}}
\def\ba {\begin{eqnarray}}
\def\ea {\end{eqnarray}}
\newtheorem{prop}{Proposition}
\newtheorem{cor}{Corollary}
\newcommand{\bra}[1]{\left(#1\right)}
\newcommand{\bras}[1]{\left[#1\right]}
\newcommand{\brac}[1]{\left\{#1\right\}}

\newcommand{\sfr}[2]{{\textstyle\frac{#1}{#2}}}

\newcommand{\lc}{\varepsilon}
\newcommand{\lb}{\{}%{\lceil}
\newcommand{\rb}{\}}%{\rfloor}
\newcommand{\E}{{\mathcal E}}
\renewcommand{\H}{{\mathcal H}}
\newcommand{\barray}{\begin{array}}
\newcommand{\earray}{\end{array}}
\newcommand{\e}{e}
\newcommand{\N}{N}

\newcommand{\del}{\nabla}
\newcommand{\hatn}{a}
\newcommand{\dotn}{\alpha}

 \newcommand{\nab}{\nabla}
\newcommand \veps {\varepsilon}
\newcommand \ep {\epsilon}
\newcommand{\LB}{\left(}
\newcommand{\RB}{\right)}
\newcommand \om {\omega}
\newcommand{\na}{\nabla}
\newcommand{\udot}{{\mathcal A}}
\begin{document}

\title{Cosmic Censorship Conjecture revisited: Covariantly}

\author{Aymen I. M. Hamid}
\email{aymanimh@gmail.com}
\affiliation{Astrophysics and Cosmology Research Unit, School of Mathematics, Statistics and Computer Science, University of KwaZulu-Natal, Private Bag X54001, Durban 4000, South Africa.}
\affiliation{Physics Department, University of Khartoum, Sudan.}
\author{Rituparno Goswami}
\email{Goswami@ukzn.ac.za}
\affiliation{Astrophysics and Cosmology Research Unit, School of Mathematics, Statistics and Computer Science, University of KwaZulu-Natal, Private Bag X54001, Durban 4000, South Africa.}
\author{Sunil D. Maharaj}
\email{Maharaj@ukzn.ac.za}
\affiliation{Astrophysics and Cosmology Research Unit, School of Mathematics, Statistics and Computer Science, University of KwaZulu-Natal, Private Bag X54001, Durban 4000, South Africa.}

\begin{abstract}
In this paper we study the dynamics of the trapped region using a frame independent 
semi-tetrad covariant formalism for general Locally Rotationally Symmetric (LRS) class II spacetimes. We covariantly prove some important geometrical 
results for the apparent horizon, and state the necessary and sufficient conditions for a singularity to be locally naked. 
These conditions bring out, for the first time in a quantitative and transparent manner, the importance of the Weyl curvature in deforming and delaying the 
trapped region during continual gravitational collapse, making the central singularity locally visible.

\end{abstract}


\maketitle
%%%%%%%%%%%%%%%%%%%%%%%%%%%%%%%%%%%%%%%%%
\section{Introduction}
%%%%%%%%%%%%%%%%%%%%%%%%%%%%%%%%%%%%%%%%%
Since Penrose proposed the famous {\it  Cosmic Censorship Conjecture (CCC)} in 1969 \cite{CCC}, stating that 
singularities observable from the outside will never arise in generic gravitational collapse which starts from a perfectly reasonable 
nonsingular initial state, there have been numerous attempts towards validating this conjecture  by means of a rigorous mathematical proof. 
However, this conjecture remains unproved, and it has been recognised as one of the most important open problems in gravitational physics. 
The key point here is that the validity of this conjecture will confirm the already widely accepted 
and applied theory of black hole dynamics, which has considerable amount of astrophysical applications. 
On the other hand, it's overturn will throw the black hole dynamics into serious doubt. This is because most of the important fundamental global theorems 
in black hole physics assume that the spacetime manifold is {\it future asymptotically predictable}. In other words this condition ensures that there should be 
no singularity to the future of the partial Cauchy surface which is `naked' or visible from the future null infinity \cite{HE}.

Although no conclusive proof or disproof of CCC could be formulated, the quest gave rise to a number 
of counter examples which showed there are shell focusing naked singularities occurring at the centre of spherically symmetric dust, perfect fluids or radiation shells
(see for example ~\cite{Goswami:2006ph, Joshibook1} and the references therein).
We can, in principle, rule out these naked singularities by stating that dust or perfect fluids are not really 
`fundamental' forms of matter field, as their properties are not derived from a `proper' Lagrangian. 
However, if the cosmic censorship is to be established as a rigorous mathematical theorem, this objection 
has to be made precise in terms of a clear and simple restriction on the energy momentum tensor of the matter field. This is necessary because 
in the above mentioned examples, the matter satisfies physically reasonable conditions such as the energy conditions or a 
well posed initial value formulation for the Einstein field equations. Also, these forms of matter are widely used in discussing 
astrophysical processes, such as collapsing stars.

Extensive studies of various dynamical collapse models for a wide range of matter fields, mainly spherically symmetric, continued over the past two decades, investigating the final outcome of gravitational collapse (refer to \cite{Joshibook2} for a detailed analysis on this subject). 
The generic conclusions which emerged from these studies were extraordinary as they conclusively indicated that while the collapse always produces curvature
 generated fireballs characterised by diverging densities and curvatures, trapped surfaces may not develop early enough to always shield this process from an outside observer. 
Not just isolated trajectories but families of non-spacelike geodesics emerge from such a naked singularity, providing a non-zero measure set of trajectories escaping away.

An obvious question of considerable interest and significance, is then the following: What are the possible physical and geometrical factors that are responsible for this delay 
in the formation of trapped regions, that cover the spacetime singularity? In other words, we would like to inquire about the physical and geometrical effects operating 
during the continual collapse of a massive matter cloud that lead to the formation of a locally naked singularity rather than a black hole, or vice versa. Such an investigation 
should help us in obtaining a better understanding of the physics of black hole or naked singularity formation in gravitational collapse. Towards this end, the pioneering work was done by 
Joshi et al \cite{Joshi:2001xi} followed by \cite{Joshi:2004tb}, which showed that spacetime shear plays a crucial role in determining the 
end state of continual collapse. The important insights that emerged from these investigations were that there exists a remarkable connection between spacetime shear and inhomogeneity of collapsing matter cloud that can distort the geometry of the trapped region is such a way that the central singularity can be locally naked. 

We continue with this investigation to obtain more transparent physical picture of the problem; in this paper we study the dynamics of the trapped region using the frame independent 
semi-tetrad covariant formalism for general Locally Rotationally Symmetric (LRS) class II spacetimes (of which spherical symmetry is a subclass) \cite{EllisLRS}. We write down the field equations 
for the LRS II spacetimes as propagation, evolution and constraint equations in terms of different covariant scalars that have well defined physical and geometrical interpretations. We deduce the equations of null geodesics for these spacetimes in terms of these scalars, and we find the equation of the apparent horizon (the boundary of the trapped region) where the 
expansion of the null geodesics vanishes. We covariantly prove some geometric results for the apparent horizon and state the necessary and sufficient conditions for a singularity to be locally naked. These conditions bring out for the first time, in a quantitative and transparent manner, the importance of the  Weyl curvature in deforming and delaying the trapped region to make the central singularity locally naked. 

As we know the Weyl tensor, which is the trace-free part of the Riemann curvature tensor, gives the measure of the `pure' geometrical effect on the spacetime manifold, as this tensor can be non-zero even in the absence of any matter field. The Weyl tensor depicts the tidal forces experienced by the test particles, resulting in volume distortion and generating the spacetime shear. In fact, it is well known that the variation of acceleration vector along with  the electric part of the Weyl tensor act as a source for the shear evolution equation in general relativity. Hence, in general, in a spacetime with non-zero electric Weyl, shear will be generated even if it is zero at a given epoch.
The electric part of Weyl tensor also gives a measure of gravitational wave propagation. Four dimensional spacetimes with vanishing Weyl tensor are conformally flat. We rigorously
show that for such spacetimes, a collapsing perfect fluid necessarily ends up in a black hole end state as the singularity is always hidden within the horizon. This then relates 
conformal flatness with local visibility (or otherwise) of a spacetime singularity.

The paper is organised as follows: In the next section we provide a brief description of the semi-tetrad 1+3 and 1+1+2 formalisms, and define the covariant kinematical and dynamical variables that have well defined geometrical and physical significance. In section 3, we use these variables to write down the field equations for LRS-II spacetimes. In section 4, we deduce the equations of null geodesics and define the expressions for expansion, shear etc., for the null congruence in the two dimensional null screen space. In section 5, we derive the 
equation for apparent horizon, which is the boundary of the trapped region in terms of these covariant variables. This then gives a local frame independent description of the horizon, 
and we prove some important covariant results for spherical collapsing shells crossing the horizon ({\it i.e.} getting trapped). 
In section 6, we give the necessary and sufficient conditions for a spacetime singularity to be locally naked. Finally in the last section we use this result to establish the 
nature of the singularity which develops as the end state of gravitational collapse, for some special cases.

Unless otherwise specified, we use natural units ($c=8\pi G=1$) throughout this paper, Latin indices run from 0 to 3. 
The symbol $\nabla$ represents the usual covariant derivative and $\partial$ corresponds to partial differentiation. 
We use the $(-,+,+,+)$ signature and the Riemann tensor is defined by
\begin{equation}
R^{a}{}_{bcd}=\Gamma^a{}_{bd,c}-\Gamma^a{}_{bc,d}+ \Gamma^e{}_{bd}\Gamma^a{}_{ce}-\Gamma^e{}_{bc}\Gamma^a{}_{de}\;,
\end{equation}
where the $\Gamma^a{}_{bd}$ are the Christoffel symbols (i.e. symmetric in the lower indices), defined by
\begin{equation}
\Gamma^{a}{}_{bd}=\frac{1}{2}g^{ae}
\left(g_{be,d}+g_{ed,b}-g_{bd,e}\right)\;.
\end{equation}
The Ricci tensor is obtained by contracting the {\em first} and the {\em third} indices
\begin{equation}\label{Ricci}
R_{ab}=g^{cd}R_{cadb}\;.
\end{equation}
The symmetrisation and the antisymmetrisation  over the indexes of a tensor are defined as 
\begin{equation}
T_{(a b)}= \frac{1}{2}\left(T_{a b}+T_{b a}\right)\;,\qquad T_{[a b]}= \frac{1}{2}\left(T_{a b}-T_{b a}\right)\,.
\end{equation}
The Hilbert--Einstein action in the presence of matter is given by
\begin{equation}
{\cal S}=\frac12\int d^4x \sqrt{-g}\left[R-2\Lambda-2{\cal L}_m \right]\;,
\end{equation}
variation of which gives the Einstein's field equations as
\be
G_{ab}+\Lambda g_{ab}=T_{ab}\;.
\ee

%%%%%%%%%%%%%%%%%%%%%%%%%%%%%%%%%%%%%%%%%
\section{Semi-tetrad covariant formalisms}
%%%%%%%%%%%%%%%%%%%%%%%%%%%%%%%%%%%%%%%%%

Spacetimes can be described using tetrad formalisms or metric (or coordinate) based approaches. 
The tetrad formalisms range from the Newman-Penrose null tetrad method, 3+1 ADM decomposition, 1+3 covariant approach developed by Ehlers and Ellis 
to 1+1+2 covariant formalism. These include either a full tetrad approach or a `partial' covariant approach where only one or two tetrad vectors are chosen. In this section
we give a brief review of the last two formalisms mentioned above.

\subsection{1+3 Covariant formalism}

This formalism \cite{EllisCovariant} is based on a local 1+3 threading of the spacetime manifold 
with respect to a timelike congruence, such that spacetime is locally decomposed into space and time.
The 1+3 formalism has been a useful tool for understanding many geometrical and physical aspects of relativistic fluid flows, both in non-linear 
GR studies or in the gauge invariant, covariant perturbation formalism \cite{Ellisbook}.

In this approach we must first define a
time-like congruence with a unit tangent vector $u^a$. The natural
choice of this vector in our case will be the
tangent to the matter flow lines. Then the spacetime is split
locally in the form $R\otimes V$ where  $R$ denotes the worldline
along $u^a$ and  $V$ is the 3-space perpendicular to $u^a$. Any
vector $X^a$  can then be projected on the 3-space by the projection
tensor $h^a{}_b=g^a{}_b+u^au_b$.  The choice of the timelike vector naturally defines 
two derivatives: the vector $ u^{a}$ is used to define the \textit{covariant
time derivative} along the observers' worldlines (denoted by a dot)
for any tensor $ S^{a..b}{}_{c..d}$, given by 
\be
\dot{S}^{a..b}{}_{c..d}{} = u^{e} \nab_{e} {S}^{a..b}{}_{c..d} 
\ee
and the tensor $ h_{ab} $ is used to define the fully orthogonally
\textit{projected covariant derivative} $D$ for any tensor $
S^{a..b}{}_{c..d} $: 
\be D_{e}S^{a..b}{}_{c..d}{} = h^a{}_f
h^p{}_c...h^b{}_g h^q{}_d h^r{}_e \nab_{r} {S}^{f..g}{}_{p..q}\;,
\ee 
with total projection on all the free indices.  Angle brackets
denote orthogonal projections of vectors, and the orthogonally
\textit{projected symmetric trace-free} PSTF part of tensors: 
\be
V^{\la a \ra} = h^{a}{}_{b}V^{b}~, ~ S^{\la ab \ra} = \left[
h^{(a}{}_c {} h^{b)}{}_d - \frac{1}{3} h^{ab}h_{cd}\right] S^{cd}\;.
\label{PSTF} 
\ee 
This splitting of spacetime also naturally defines
the 3-volume element 
\be \ep_{a b c}=-\sqrt{|g|}\delta^0_{\left[ a
\right. }\delta^1_b\delta^2_c\delta^3_{\left. d \right] }u^d\;,
\label{eps1} 
\ee 
with the following identities 
\be 
\ep_{a b c}\ep^{d
e f}=3!h^d_{\left[ a \right. }h^e_bh^f_{\left. c \right] }\;\;;\;\;
\ep_{a b c}\ep^{d e c}=2!h^d_{\left[ a \right. }h^e_{\left. b
\right] } \label{eps2}. 
\ee
The covariant derivative of the time-like vector $u^a$ can now be
decomposed into the irreducible part as 
\be
\nabla_au_b=-A_au_b+\frac13h_{ab}\Theta+\sigma_{ab}+\ep_{a b
c}\om^c\;, 
\ee 
where $A_a=\dot{u_a}$ is the acceleration,
$\Theta=D_au^a$ is the expansion, $\sigma_{ab}=D_{\la a}u_{b \ra}$
is the shear tensor and $w^a=\ep^{a b c}D_bu_c$ is the vorticity
vector. Similarly the Weyl curvature tensor can be decomposed
irreducibly into the Gravito-Electric and Gravito-Magnetic parts as
\be 
E_{ab}=C_{abcd}u^cu^d=E_{\la ab\ra}\;;\;
H_{ab}=\frac12\ep_{acd}C^{cd}{}_{be}u^e=H_{\la ab\ra}\;, 
\ee 
which allows for a covariant description of tidal forces and gravitational
radiation. The energy momentum tensor for a general matter field can 
be similarly decomposed as follows:
\be
T_{ab}=\mu u_au_b+q_au_b+q_bu_a+ph_{ab}+\pi_{ab}\;,
\ee
where $\mu=T_{ab}u^au^b$ is the energy density, $p=(1/3 )h^{ab}T_{ab}$ is the isotropic pressure, $q_a=q_{\la a\ra}=-h^{c}{}_aT_{cd}u^d$ is the 3-vector defining 
the heat flux and $\pi_{ab}=\pi_{\la ab\ra}$ is the anisotropic stress.

\subsection{1+1+2 Covariant formalism}

A natural extension to the 1+3 formalism, which is optimised for spacetimes having a preferred spatial direction (for example spherical symmetry), is the 1+1+2 formalism developed recently by Clarkson and Barrett and it has been used extensively to study perturbations of black holes \cite{Clarkson:2002jz,Betschart:2004uu,Clarkson:2007yp}. In this formalism we first proceed with the same split of the 1+3 approach followed by another split along a preferred spatial direction. This allows us to derive a set of covariant scalar variables which are more advantageous to treat systems with one preferred direction. For example in spherically symmetric systems the governing field equations in the 1+1+2  approach are scalar equations and are much simpler than the ones of the 1+3 formalism which are in general tensorial.  
 
Hence in this approach we choose a further preferred vector field $\e^a$ which performs additional 
slicing of the `3-space'. This new vector field has to be orthogonal to $u^a$ such that it satisfies
$\e^a\e_a=1,~u^a\e_a=0$. The 1+3 projection tensor
$h_{a}^{~b}\equiv g_{a}^{~b}+u_au^b$ combined with $e^a$ defines a new projection tensor $N_{ab}$, \be \N_a^{~b}\equiv
h_a^{~b}-\e_a\e^b=g_{a}^{~b}+u_au^b-\e_a\e^b\ , \ee which projects
vectors orthogonal to $\e^a$ and $u^a$
($\e^a\N_{ab}=0=u^a\N_{ab}$) onto a 2-surface which is defined as the sheet ($\N_a^{~a}=2$). The volume element of this 
2-surface is then Levi-Civita 2-tensor, derived from the volume
element $\ep_{abc}$ for the observers'
rest spaces by \be \lc_{ab}\equiv\ep_{abc}\e^c = u^d\eta_{dabc}e^c\
;\qquad \lc_{ab}\e^b=0=\lc_{(ab)}\ . \ee
Any 3-vector $\psi^a$ can now be irreducibly split into a scalar, $\Psi$, which
is the part of the vector parallel to $\e^a$, and a vector,
$\Psi^a$, lying in the 2-surface orthogonal to $\e^a$: \ba 
\psi^a&=&\Psi\e^a+\Psi^{a},~~~\mbox{where}~~~\Psi\equiv \psi_a\e^a\
,\nonumber\\&&~~~\mbox{and}~~~\Psi^{a}\equiv \N^{ab}\psi_b\equiv \psi^{\bar a}\label{psia},
\ea where the bar over the index denotes projection with $\N_{ab}$.
Similarly, we can do the same for any tensor,
$\psi_{ab}$, as follows:
 \be
\psi_{ab}=\psi_{\langle
ab\rangle}=\Psi\bra{\e_a\e_b-\sfr12\N_{ab}}+2\Psi_{(a}\e_{b)}+\Psi_{{ab}}\
, \label{tensor-decomp} 
\ee where
 \ba
\Psi&\equiv &\e^a\e^b\psi_{ab}=-\N^{ab}\psi_{ab}\ ,\nonumber\\
\Psi_a&\equiv &\N_a^{~b}\e^c\psi_{bc}=\Psi_{\bar a}\ ,\nonumber\\
\Psi_{ab}&\equiv &
\bra{\N_{(a}^{~~c}\N_{b)}^{~~d}-\sfr{1}{2}\N_{ab}\N^{cd}}\psi_{cd}
\equiv \Psi_{\lb ab\rb}\ . \label{PSTF-TT}
\ea 
The curly brackets denote the PSTF tensors on the 2-sheets. Apart from the `{\it time}' (dot) derivative, of an object, we now introduce two new derivatives, 
which $ e^{a} $ defines, for any object $ \psi_{a...b}{}^{c...d}  $: 
\ba
\hat{\psi}_{a..b}{}^{c..d} &\equiv & e^{f}D_{f}\psi_{a..b}{}^{c..d}~, 
\\
\delta_f\psi_{a..b}{}^{c..d} &\equiv & N_{a}{}^{f}...N_{b}{}^gN_{h}{}^{c}..
N_{i}{}^{d}N_f{}^jD_j\psi_{f..g}{}^{i..j}\;.
\ea 
The hat-derivative is the derivative along the $e^a$ vector-field in the surfaces orthogonal to $ u^{a} $. The $\delta$ -derivative is the projected derivative onto the sheet, with the projection on every free index.

We can now split the usual 1+3 kinematical and Weyl quantities into the irreducible
set $\{\Theta,\udot,\Omega,\Sigma,{\cal E},{\cal H},\udot^a,\Sigma^a,{\cal
E}^a,{\cal H}^a,\Sigma_{ab},{\cal E}_{ab},{\cal H}_{ab}\}$ using
(\ref{psia}) and (\ref{tensor-decomp}) as follows \cite{Clarkson:2007yp}:
\ba
\uudot^a&=&\udot \e^a+\udot^a,\\
\omega^a&=&\Omega \e^a+\Omega^a,\\
\sigma_{ab}&=&\Sigma\bra{\e_a\e_b-\sfr{1}{2}\N_{ab}}+2\Sigma_{(a}\e_{b)}+\Sigma_{ab},\\
E_{ab}&=&{\cal E}\bra{\e_a\e_b-\sfr{1}{2}\N_{ab}}+2{\cal E}_{(a}\e_{b)}+{\cal E}_{ab},\\
H_{ab}&=&{\cal H}\bra{\e_a\e_b-\sfr{1}{2}\N_{ab}}+2{\cal H}_{(a}\e_{b)}+{\cal
H}_{ab}.
\ea
The shear scalar, $\sigma$, for example, may be expressed in the
form
\be
\sigma^2 \equiv \sfr{1}{2}\sigma_{ab}\sigma^{ab} = \sfr{3}{4}\Sigma^2 +
\Sigma_a\Sigma^a + \sfr{1}{2}\Sigma_{ab}\Sigma^{ab}.
\ee
Similarly we may split the fluid variables $q^a$ and $\pi_{ab}$,
\ba
q^a&=&Q \e^a+Q^a,\\
\pi_{ab}&=&\Pi\bra{\e_a\e_b-\sfr{1}{2}\N_{ab}}+2\Pi_{(a}\e_{b)}+\Pi_{ab}.
\ea
We are now able to decompose the covariant derivative of $e^a$ in the direction orthogonal to $u^a$ into it's irreducible parts giving 
\be 
{\rm D}_{a}e_{b} = e_{a}a_{b} + \frac{1}{2}\phi N_{ab} + 
\xi\epsilon_{ab} + \zeta_{ab}~, 
\ee
where 
\ba 
a_{a} &\equiv & e^{c}{\rm D}_{c}e_{a} = \hat{e}_{a}~, \\ 
\phi &\equiv & \delta_ae^a~, \\  \xi &\equiv & \frac{1}{2} 
\epsilon^{ab}\delta_{a}e_{b}~, \\ 
\zeta_{ab} &\equiv & \delta_{\lb a}e_{b \rb }~.
\ea
We see that on the 3-space,  moving along the preferred vector $ e^{a} $, $\phi$ represents the \textit{expansion of the sheet},  $\zeta_{ab}$ is the \textit{shear of $e^{a}$} (i.e. the distortion of the sheet) and $a^{a}$ its \textit{acceleration}. We can also interpret $\xi$ as the \textit{vorticity} associated with $e^{a}$ so that it is a representation of the ``twisting'' or rotation of the sheet. 

The Ricci identities for $\e_a$ is given by:
\be
R_{abc}\equiv2\del_{[a}\del_{b]}\e_c-R_{abcd}\e^d=0,\label{ricci}
\ee
where $R_{abcd}$ is the Riemann curvature tensor.
And the full covariant derivative of $\e_a$ and $u_a$ is now written as:
\ba\label{del_e}
\,\,\, \del_a\e_b&=&-\udot u_au_b-u_a\dotn_b +\bra{\Sigma+\sfr13\Theta}\e_a u_b+\xi\lc_{ab}\nonumber \\
&& +\bra{\Sigma_a-\lc_{ac}\Omega^c}u_b +\e_a\hatn_b
+\sfr{1}{2}\phi\N_{ab},
\ea
\ba\label{del_u}
\del_au_b&=&
-u_a\bra{\udot\e_b+\udot_b}+\e_a\e_b\bra{\sfr13\Theta+\Sigma}+\Omega\lc_{ab}\nonumber \\
&&+\e_a\bra{\Sigma_b+\lc_{bc}\Omega^c}+\bra{\Sigma_a-\lc_{ac}\Omega^c}\e_b\nonumber\\
&&+\N_{ab}\bra{\sfr13\Theta-\sfr12\Sigma}+
\Sigma_{ab},
\ea
We also write one more useful relation
\be
\hat{u}_a =
\bra{\sfr13\Theta+\Sigma}\e_a+\Sigma_a+\lc_{ab}\Omega^b.
\label{uhat}
\ee
\section{LRS-II spacetimes}

A spacetime manifold $(\mathcal{M},g)$ is called  {\it locally isotropic}, if every point $p\in (\mathcal{M},g)$ has continuous non-trivial isotropy group. When this group 
consists of spatial rotations the spacetime is called {\it locally rotationally symmetric} (LRS) \cite{EllisLRS}. These spacetimes exhibit locally (at each point) a unique preferred spatial direction, covariantly defined (for example, by the vorticity vector field or a non-vanishing acceleration of the matter fluids, etc.). The 1+1+2 formalism is therefore ideally suited for covariant description of these spacetimes. The preferred spatial direction in the LRS spacetimes constitutes a local axis of symmetry and in this case $ e^{a} $ is just a vector pointing along the axis of symmetry. Since LRS spacetimes are isotropic about this axis, \textit{all} 2-vectors and 2-tensors vanish, so that there are no preferred directions in the sheet. Thus, all the non-zero 1+1+2 variables are covariantly defined scalars. The variables 
$\brac{\udot, \Theta,\phi, \xi, \Sigma,\Omega, \E, \H, \mu, p, \Pi, Q }$, fully describe LRS spacetimes, and are what is solved for in the 1+1+2 approach. 

Within the LRS cases is the LRS-II class that admits spherically symmetric solutions and is free of rotation, thus allowing for the vanishing of the variables $\Omega$, $ \xi $ and $ \H $. The set of quantities that fully describe LRS class II spacetimes are $\brac{\udot, \Theta,\phi, \Sigma,\E, \mu, p, \Pi, Q }$.
It was shown that the most general metric for LRS II can be written as \cite{Stewart-Ellis}
\ba
   ds^2 &= &-A^{2}(t,\chi)\,dt^2+B^2(t,\chi)\,d\chi^2 \nonumber\\
   && +C^2(t,\chi)\,[\,dy^2+D^2(y,k)\,dz^2\,] \;,\label{LRSds}
\ea
where $t$ and $\chi$ are the affine parameters along the integral curves of $u^a$ and $e^a$ respectively, and $k=(1,0,-1)$ describes the closed, flat or open geometry 
of the 2-sheets respectively.

We now have all the tools to derive the propagation and the evolution equations for the LRS-II
variables (for more details see \cite{Clarkson:2007yp}).  These equations are obtained by the Ricci identities of the vectors $u^a$ and $e^a$ and the doubly contracted Bianchi identities.

\textit{Propagation}:
\ba
\hat\phi  &=&-\sfr12\phi^2+\bra{\sfr13\Theta+\Sigma}\bra{\sfr23\Theta-\Sigma}
    \nonumber\\&&-\sfr23\bra{\mu+\Lambda}
    -\E -\sfr12\Pi,\,\label{hatphinl}
\\  
\hat\Sigma-\sfr23\hat\Theta&=&-\sfr32\phi\Sigma-Q\
,\label{Sigthetahat}
 \\  
\hat\E-\sfr13\hat\mu+\sfr12\hat\Pi&=&
    -\sfr32\phi\bra{\E+\sfr12\Pi}
    +\bra{\sfr12\Sigma-\sfr13\Theta}Q.\;
\ea

\textit{Evolution}:
\ba
   \dot\phi &=& -\bra{\Sigma-\sfr23\Theta}\bra{\udot-\sfr12\phi}
+Q\ , \label{phidot}
\\   
\dot\Sigma-\sfr23\dot\Theta
&=&
-\udot\phi+2\bra{\sfr13\Theta-\sfr12\Sigma}^2\nonumber\\
        &&+\sfr13\bra{\mu+3p-2\Lambda}-\E+\sfr12\Pi\, ,\label{Sigthetadot}
\\  
\dot\E -\sfr13\dot \mu+\sfr12\dot\Pi &=&
    +\bra{\sfr32\Sigma-\Theta}\E
    +\sfr14\bra{\Sigma-\sfr23\Theta}\Pi\nonumber\\
    &&+\sfr12\phi Q-\sfr12\bra{\mu+p}\bra{\Sigma-\sfr23\Theta}\ . \label{edot}
\ea

\textit{Propagation/evolution}:
\ba
   \hat\udot-\dot\Theta&=&-\bra{\udot+\phi}\udot+\sfr13\Theta^2
    +\sfr32\Sigma^2 \nonumber\\
    &&+\sfr12\bra{\mu+3p-2\Lambda}\ ,\label{Raychaudhuri}
\\
    \dot\mu+\hat Q&=&-\Theta\bra{\mu+p}-\bra{\phi+2\udot}Q -
    \sfr32\Sigma\Pi,\,
\\    \label{Qhat}
\dot Q+\hat
p+\hat\Pi&=&-\bra{\sfr32\phi+\udot}\Pi-\bra{\sfr43\Theta+\Sigma} Q\nonumber\\
    &&-\bra{\mu+p}\udot\ .
\ea
The 3-Ricci scalar of the spacelike 3-space orthogonal to $u^a$ can be expressed as
\be
^3R =-2\bras{\hat \phi +\sfr34\phi^2 -K}\;,\label{3sca}
\ee
where $K$ is the Gaussian curvature of the 2-sheet defined by
$^2R_{ab}=KN_{ab}$. In terms of the covariant scalars we can write the Gaussian curvature $K$ as
\be
K = \sfr13\bra{\mu+\Lambda}-\E-\sfr12\Pi +\sfr14\phi^2
-\bra{\sfr13\Theta-\sfr12\Sigma}^2\ . \label{gauss}
\ee
Finally the evolution and propagation equations for the Gaussian curvature $K$ are
\ba
\dot K &=& -\bra{\sfr23\Theta-\Sigma}K\ , \label{Kdot}\\
\hat K &=& -\phi K\ . \label{Khat}
\ea
\section{Null geodesics in LRS-II spacetimes}
In this section, we derive the equation for null geodesics in LRS-II spacetimes and investigate the geometry of these null congruences. 
Null geodesics (light rays) are characterised by the curves $x^{a}(\nu)$ on $(\mathcal{M},g)$, where $\nu$ is an affine parameter along the geodesics. 
The tangent to these curves is defined by
 \be
k^{a} = \frac{dx^{a}}{d\nu}(\nu)
\ee
where $k^a$ is a null vector obeying
\be
k^{a}k_{a} = 0~.\label{nullvector}
\ee
Also, since the tangent vector to the geodesic is parallely propagated to itself, we can write
\be\label{geo-k}
k^{b}\na_{b}k^{a} =\frac{\delta k^{a}}{\delta \nu} = 0~, 
\ee
where $\frac{\delta }{\delta \nu} = k^{b}\na_{b}$ as the derivative along the ray. 
In the usual $1+3$ decomposition of null geodesics, we define 
the unit spacial vector $n_a$ as
\be
n^{a}n_{a} = 1~, ~~ n^{a}u_{a} = 0\;.
\ee
The null vector $k^a$  can now be split in the usual way
\be
k^a=E(u^a+n^a)\,,
\ee
The
$1+1+2$  split of $n^a$ can then be performed, such that \cite{deSwardt:2010nf, Nzioki:2010nj}
\be 
k^{a} = E(u^{a} + \kappa e^{a} + \kappa^{a})\,,\label{splitNull} 
\ee
where $E \equiv -u_{a}k^{a}$ can be interpreted as the energy associated with the ray,
$\kappa\equiv k^ae_a$ is the {magnitude of the component along the preferred spatial direction,} and $\kappa^a$ is the 
component lying on the 2-sheet.  

At this point let us define the notion of locally {\it outgoing} and {\it incoming} null geodesics with respect to the preferred spatial direction. Consider 
any open subset ${\mathcal S}$ of $(\mathcal{M},g)$ and let $x^{a}(\nu)$ be a null geodesic in ${\mathcal S}$. Let $k^a$ be the tangent to this geodesic. 
If $e^ak_a>0$ in ${\mathcal S}$ then the geodesic is considered to be outgoing with respect to the preferred direction in ${\mathcal S}$. Similarly 
$e^ak_a<0$ denotes an incoming geodesic. This can also be explained in terms of the local co-ordinates.
Let $p_1$ and $p_2$ be two points 
on $x^{a}(\nu)$ such that $p_2$ is in the causal future of $p_1$. Both these points can be labelled by the values of `$t$' and `$\chi$' 
(which are the affine parameters 
on the integral curves of the vectors $u^a$ and $e^a$ respectively) and the local coordinates on the 2-sheets. Let the values of $\chi$ at these points 
be $\chi_1$ and $\chi_2$ respectively. If $\chi_2>\chi_1$ then the geodesic is considered to be outgoing and if $\chi_2<\chi_1$ then the geodesic is considered to be 
incoming (with respect to the preferred direction).

%%%%%%%%%%%%%%%%%%%%%%%%%%%%%%%%%%%%%%%%%%%%%%%%%%%%%%%%%%%%%%%%%%%%%%%%
\subsection{The propagation equations for the null geodesics}
%%%%%%%%%%%%%%%%%%%%%%%%%%%%%%%%%%%%%%%%%%%%%%%%%%%%%%%%%%%%%%%%%%%%%%%%
The propagation equations for the energy $E$ and the component $\kappa$ in a LRS-II spacetime (where the sheet component $\kappa^{a}$ vanishes)
can be derived by substituting 
(\ref{splitNull}) into (\ref{geo-k}), and projecting the expression along the timelike direction ($u^{a}$)
and along the radial direction ($e^{a}$) \cite{deSwardt:2010nf,Nzioki:2010nj}
%Each ray is parameterized by the affine parameter, $\nu$, so that the geodesic condition can be written as :
%\be 
%k^{b} \nab_{b} k^{a} = \frac{\delta k^{a}}{\delta \nu} = (k^{a})' = 0~,\label{NullCond1} 
%\ee
%where the prime derivative ($'$) denotes change along the ray (i.e. with respect to the affine parameter $\nu$). 
\be
\frac{\delta E}{\delta \nu} =E' = - E^{2}\kappa\udot- \sfr32\Sigma \kappa^{2} E^{2} -E^{2} \bra{\sfr13\Theta - \sfr12\Sigma}~, \label{Eprime1}
\ee
\be
\frac{\delta \kappa}{\delta \nu} =\kappa{}' = E \bra{1-\kappa^{2}} \bra{\sfr12 \phi - \udot -\sfr32\Sigma}~\label{Kprime1}.
\ee
We have used the following properties
\ba
k^{b}u_{b} &=& -E,~~~ k^{b}e_{b}=E \kappa, ~~~ N^{a}{}_{b} k^{b}= E \kappa^{a},\nonumber\\
\veps^{a}{}_{b}k^{b} &=& E \veps^{a}{}_{b}\kappa^{b},~~~ u_{a}\kappa^{a} = 0,~~~ e_{a}\kappa^{a} = 0\;,
\ea
as well as \cite{Nzioki:2010nj}
\ba
u_{a}' &=& E\udot e_{a} + E\kappa \LB \sfr13\Theta + \Sigma \RB
e_{a}  \\ \nonumber
&& + E\LB \sfr13\Theta - \sfr12\Sigma \RB\kappa_a + E\Omega \veps_{ab}\kappa^b ~,\label{uPrime} \\
e_{a}' &=& E\udot u_{a} + E\kappa\LB \Sigma + \sfr13\Theta \RB u_{a} +
\sfr12 E\phi\kappa_a + E\xi\veps_{ab}\kappa^b~,\nonumber \label{ePrime}
\ea
which are obtained from (\ref{del_e}\,,\,\ref{del_u}) with the definition of prime introduced above. 
For null rays along the preferred spatial direction we have $\kappa=\pm 1$ (denoting the outgoing and incoming geodesics). Then we can easily 
see that the equation (\ref{Kprime1}) is satisfied identically and (\ref{Eprime1}) simplifies to
\be
\frac{\delta E}{\delta \nu} =E' = \mp E^{2}\udot - E^{2} \bra{\Sigma+\sfr13\Theta}~.\label{Eprime2}
\ee

\subsection{The Screen-Space}
As we have already seen,  for LRS-II spacetimes the outgoing null vector is defined as 
\be
k^a=E\bra{u^a+e^a}.
\ee
Since the hypersurface orthogonal to null vector $k^a$, contains $k^a$ and hence the projection onto a locally orthogonal space now has to be defined differently.
Let us now define the projection tensor $\tilde{h}_{ab}$, which projects tensors and vectors into the 2-D 
screen space orthogonal to $k^a$, as \cite{deSwardt:2010nf}
\be
\tilde{h}_{ab}\equiv g_{ab}+2k_{(a}l_{b)}, \,\,\tilde{h}^a_a=2,\,\,\tilde{h}_{ac}\tilde{h}^c_b=\tilde{h}_{ab},\,\,\tilde{h}_{ab}k^b=0,\label{ht}
\ee
where $l_a$ is null ingoing geodesic that obeys
\be
l^al_a=0,\,\,k^al_a=-1\,\, \text{and} \,\,\frac{\delta l^a}{\delta \nu}=k^b\nabla_bl^a=0.
\ee
Using these definitions, the general form of $l^a$ can be written as:
\be
l^a=\frac{1}{2E}\bra{u^a-e^a}\label{l},
\ee
and substituting (\ref{l}) into (\ref{ht}) the screen-space projection tensor is obtained as
\be
\tilde{h}_{ab}=g_{ab}+u_au_b-e_ae_b.
\ee
It is interesting to note that although defined differently, we automatically have 
\be
\tilde{h}_{ab}=N_{ab}
\ee
 An expression for any vector or tensor lying on the 2-D surface can be obtained by
\be
\tilde{V}^a=\tilde{h}^{a }_{~b}V_b,\,\,\,\tilde{T}^{a\dots c}_{~~~~~b\dots d}=
\tilde{h}^a_{~e}\tilde{h}^f_{~b}\dots \tilde{h}^h_{~d}T^{e\dots g}_{~~~~~f\dots h}.
\ee
For completeness, we will write here the full 1+3 decomposition of the covariant derivative of the null vector $k^a$ 
for a general spacetime \cite{deSwardt:2010nf}
\be
\nabla_bk_a=\sfr12 \tilde{h}_{ab}\tilde{\Theta}_{out}+\tilde{\sigma}_{ab}+\tilde{\omega}_{ab}+\tilde{X}_ak_b+\tilde{Y}_bk_a
+\lambda k_ak_b,
\ee
where 
\be
\tilde{X}_a=\frac1Ee^d\nabla_dk_a,\;\tilde{Y}_a=\frac1Ee^d\nabla_ak_d,\;\lambda=-\frac{1}{E^2}e^ce^d\nabla_dk_c,
\ee
and  $\tilde{\Theta}_{out}$, $\tilde{\sigma}_{ab}$, $\tilde{\omega}_{ab}$ represent the expansion, shear and vorticity 
of the outgoing null congruence respectively. A similar decomposition can be done for the incoming null geodesic $l^a$.

\section{Apparent Horizon in spherically symmetric spacetimes}

As we now have a complete picture of the equations governing the geometry of null geodesics in LRS-II spacetimes,  we will use these results in 
this section to derive some important propositions regarding the apparent (or cosmological) horizons.
Henceforth we will only consider the class of spherically symmetric spacetimes which belongs to the LRS-II class 
with an extra condition of positivity of the Gaussian curvature of the 2-sheets ($K>0$). 

Let us briefly discuss the concept of a {\em closed trapped surface} for a spherically symmetric spacetime. As described in \cite{HE}, 
we will consider a spherical emitter, surrounding a massive body, emitting a flash of light. In the normal circumstances, by Huygen's construction, 
there will be outgoing and incoming spherical wavefronts and the surface area of the outgoing wavefronts will be greater than the 
emitting sphere while that of the incoming wavefront will be less than the emitting sphere. In other words, in a normal situation, 
the volume expansion of the outgoing null congruence orthogonal to the sphere is always positive ($\tilde{\Theta}_{out}>0$) while that of the incoming congruence is always 
negative ($\tilde{\Theta}_{in}<0$). However, if sufficiently large amount of matter is present within the emitting sphere, the surface areas of {\it both} 
incoming and outgoing wavefronts will be less than that of the emitting sphere. The surface of the emitting sphere is then said to be a {\em closed trapped surface}. 
In other words the volume expansion of the outgoing null congruence orthogonal to a closed trapped surface is negative. The collection of all closed trapped surfaces
in a four dimensional spacetime manifold constitutes a {\it trapped region}. The boundary of the trapped region is called the {\it apparent horizon} where the 
volume expansion of the outgoing null congruence vanishes ($\tilde{\Theta}_{out}=0$). For expanding cosmologies (like de-Sitter universe) we can similarly define 
the {\it cosmological horizon} where ($\tilde{\Theta}_{in}=0$). For a detailed discussion on trapped surfaces and black holes we refer to \cite{senovilla} (and the 
references therein).
\begin{prop}
For any spherically symmetric spacetime $(\mathcal{M},g)$ that allows a local 1+1+2 splitting, the apparent horizon is described by the 
curve $ \bra{\sfr23\Theta-\Sigma+\phi}=0 $, while the cosmological horizon is described by $ \bra{\sfr23\Theta-\Sigma-\phi}=0 $, in the local $[u,e]$ plane.
\end{prop}
\begin{proof}
We know, by definition, $\tilde{\sigma}^a_{~a}=0$, $e^a\tilde{\sigma}_{ab}$ $=0=u^a\tilde{\sigma}_{ab}$,
$e^a\tilde{\omega}_{ab}=0=u^a\tilde{\omega}_{ab}$. Also together with the properties in (\ref{ht}), we can easily 
conclude that
\ba
\tilde{\Theta}_{out}&=&\tilde{h}^{ab}\nabla_bk_a\nonumber\\
&=&EN^{ab}\nabla_a\bra{u_b+e_b}
 \label{tth}.
\ea
Now using (\ref{del_u}) and (\ref{del_e}) in (\ref{tth}) we obtain,
\be
\tilde{\Theta}_{out}=E\bra{\sfr23\Theta-\Sigma+\phi}\;.
\ee
Hence for a null congruence with non-zero energy $E$, $\tilde{\Theta}_{out}=0$ implies that $ \bra{\sfr23\Theta-\Sigma+\phi}=0 $.
Similarly we can use the decomposition of the incoming null vector $l^a$ to obtain the equation for the cosmological horizon. 
$\tilde{\Theta}_{in}=0$ will then imply $ \bra{\sfr23\Theta-\Sigma-\phi}=0 $.
 \end{proof}
\begin{prop}
For any spherically symmetric spacetime $(\mathcal{M},g)$ that allows a local 1+1+2 splitting, the gradient of the Gaussian curvature of 
the 2-sheets that intersect with the apparent (or cosmological) horizon is null.
\end{prop}
\begin{proof}
Let us calculate the quantity $\nabla_a K\nabla^aK$ for a spherically symmetric spacetime (where $K\ne 0$):
\be
\nabla_a K\nabla^aK=\bra{-u^au^b+e^ae^b}\nabla_aK\nabla_bK=-\dot{K}^2+\hat{K}^2\label{Geq}.
\ee
Now using (\ref{Kdot}) and (\ref{Khat}) in (\ref{Geq}) we get
\be
\nabla_a K\nabla^aK=\bra{\sfr23\Theta-\Sigma+\phi}\bra{\sfr23\Theta-\Sigma-\phi}K^2\;.
\ee
Hence for the 2-sheets intersecting the horizon (apparent or cosmological), the gradient of their Gaussian curvature is null.
\end{proof}
As we are considering the scenario of gravitational collapse of massive stars, henceforth we will only concentrate on the apparent horizon. We have already seen that the 
curve
\be
\Psi\equiv\frac23\Theta-\Sigma+\phi=0,
\ee
describes the apparent horizon. Let the vector $\Psi^a=\alpha u^a+\beta e^a$ be the tangent to the curve in the local $[u,e]$ plane. 
Then we must have $\Psi^a \nabla_a\Psi=0$. Since we know that $\nabla_a\Psi=-\dot{\Psi}u_a+\hat{\Psi}e_a$, we can immediately see 
the slope of the tangent to the apparent horizon on the local $[u,e]$ plane is given by $\sfr{\alpha}{\beta}=-\sfr{\hat{\Psi}}{\dot{\Psi}}$.
Now using this decomposition with the field equations (\ref{hatphinl}) to (\ref{Qhat}), we obtain
\ba\label{normal}
\nabla_a\Psi=&&\bra{\sfr13\mu+p-{\cal E}+\sfr12 \Pi-Q}u_a\nonumber\\&&+\bra{-\sfr23\mu-\sfr12\Pi-{\cal E}+Q}e_a,
\ea
and hence
\be\label{alph}
\frac{\alpha}{\beta}=\frac{\sfr23\mu+\sfr12\Pi+{\cal E}-Q}{-\sfr13\mu-p+{\cal E}-\sfr12 \Pi+Q}\;.
\ee
It is interesting to note that the matter thermodynamic quantities together with the Weyl scalar completely determine the tangent to the apparent horizon. 
We will define the apparent horizon to be locally {\it outgoing} at a point $p\in[u,e]$, if the slope of the tangent to the horizon is positive at $p$, that is 
$\sfr{\alpha}{\beta}>0$. Let the point $p$ be labelled by the values of the local coordinates ($t_0,\chi_0$) which are the affine parameters along the 
integral curves of $u^a$ and $e^a$ 
respectively. Then a locally outgoing apparent horizon at $p$ would imply that the 2-sheets (spherical shell) labelled by $\chi_0+\epsilon$ will get trapped 
later than $t=t_0$, while the 2-sheet labelled by $\chi_0$ gets trapped at $t=t_0$. Finally as we can easily see that the sign of the scalar $\Psi^a\Psi_a$, determines whether the 
curve $\Psi=0$ is timelike, spacelike or null in the $[u,e]$ plane. Hence $\sfr{\alpha^2}{\beta^2}>(<)1$ denotes the horizon to be locally timelike (spacelike). If $\sfr{\alpha^2}{\beta^2}=1$ then the horizon is null.

As an example let us consider the spherically symmetric vacuum spacetime. Then by Birkhoff's theorem the spacetime is static and hence $\Theta=\Sigma=0$ \cite{GosEllis}. 
Thus the horizon is described by the curve $\phi=0$. In this case all the matter variables vanish, we have $\sfr{\alpha}{\beta}=1$ and we can easily see that the horizon is outgoing null. This is the {\it event horizon} of the Schwarzschild spacetime. Indeed if we calculate $\phi$ in Schwarzschild coordinates we get 
\be
\phi=\frac2r\sqrt{1-\frac{2m}{r}},
\ee
and $\phi=0$ corresponds to the event horizon at $r=2m$.

\section{End state of a spherical gravitational collapse}
 
Having derived the equations that govern the dynamics of the apparent horizon in a spherically symmetric spacetime, we are now in a position 
to analyse the end state of continual gravitational collapse. Let us consider the continual collapse of a general matter cloud to a final shell-focusing singularity, where all matter shells collapse to a zero physical radius. In particular, we analyse specifically the nature of the central singularity in detail to determine when it will be covered by the horizon, and when it will be visible and causally connected to outside observers. If there are future directed families of nonspacelike curves coming out from the singularity and reaching faraway observers, then the singularity will be naked. The absence of such families will give the covered case when the result is a black hole. We specifically focus on the central singularity as it has been shown a numerous times that if all the physically reasonable energy conditions are satisfied by the collapsing matter, then the non-central singularities are always covered \cite{Joshibook2}.

Broadly, it can be stated that, if the neighbourhood of the centre gets trapped earlier than the singularity, then it is covered, otherwise it is naked with families of escaping nonspacelike future directed trajectories escaping away from it. Here we implicitly assume that the singularity curve (time taken for a spherical shell to become singular) is a non-decreasing function of the affine parameter of the integral curve of the vector $e^a$. Otherwise non-central shells will become singular before the central shell and we will have to be contents with pathologies like shell crossing singularities. 

We would like to emphasize here that we are considering the absence of shell-crossing singularities as an extra condition on the spacetime. In terms of the covariant geometrical variables, this condition is equivalent to $\hat{K}<0$ throughout the collapsing spacetime. From equation (\ref{Khat}) we can immediately see that for a collapsing shell with non-zero Gaussian curvature, $\phi>0$ ensures no shell crossing condition. In other words, the 3 dimensional expansion of the spacelike vector $e^a$ should not vanish anywhere in the collapsing spacetime.

\begin{prop}
Consider the continued collapse of a general spherically symmetric matter cloud from a regular initial epoch and obeying the physically reasonable energy conditions.
If the following conditions are satisfied:
\begin{enumerate}
\item The spacetime is free of shell crossing singularities,
\item Closed trapped surfaces exist,
\end{enumerate}
then the necessary and sufficient condition for the central singularity to be locally naked is that the slope of the tangent to the apparent horizon at the central singularity is
positive and non-spacelike ($\sfr{\alpha}{\beta}\ge1$).
\end{prop}
\begin{proof} 
Let the central singularity be denoted by ($t=t_{s_0}$, $\chi=0$) in the $[u,e]$ plane. The key point here is that there should be available untrapped region in the local neighbourhood 
of the central singularity for a null geodesic with the past end point arbitrarily near the central singularity to escape. We have assumed here that the singularity curve is a non-decreasing  function of the affine parameter of the integral curve of the vector $e^a$ (see \cite{Goswami:2002he}), and hence no other collapsing shells becomes singular before the central shell. If the apparent horizon at the central singularity is ``{\em ingoing}", that is $\sfr{\alpha}{\beta}<0$, then the neighbourhood of the centre gets trapped before the central singularity and no null geodesic from a point arbitrarily close to the central singularity can escape. Also 
if the apparent horizon is ``{\em outgoing}" but spacelike, that is $0\le\sfr{\alpha}{\beta}<1$, then any outgoing null direction from the central singularity will be necessarily within the trapped region. Hence for these cases, any null geodesic from a point arbitrarily close to the central singularity will have $\tilde\Theta_{out}<0$ and hence they will fall to the 
singularity. Therefore the necessary condition for a singularity to be  locally naked is that the slope of the tangent to the apparent horizon at the central singularity is
positive and non-spacelike ($\sfr{\alpha}{\beta}\ge1$).
Conversely, suppose there exist a family of future directed null geodesics that has escaped from the points arbitrarily close to the central singularity in the $[u,e]$ plane . Then that would imply these points are non-trapped and the slope of the apparent horizon curve at the central singularity is greater than (or equal to) the slope of these outgoing null geodesic in order for them to escape. Hence $\sfr{\alpha}{\beta}\ge1$ is the necessary and sufficient condition for the singularity to be locally naked.
 \end{proof} 
This result is interesting as it transparently explains the role of the energy momentum tensor of the collapsing matter field as well as the Weyl curvature in making a spacetime singularity locally visible. Also, as shown in \cite{psj}, if a null geodesic emerge from the singularity, then there exist families of future- directed nonspacelike curves which also necessarily escape from the same. The existence of such families is crucial to the physical visibility of the singularity. In the next proposition we show 
the crucial importance of the Weyl curvature in deforming the trapped region in such a way that the singularity becomes locally visible.
\begin{prop}
Consider the gravitational collapse of spherically symmetric perfect fluid obeying strong energy condition $\mu\ge0$ and $\mu+3p\ge0$.
If the following conditions are satisfied :
\begin{enumerate}
\item The spacetime is free of shell crossing singularities,
\item Closed trapped surfaces exist,
\item The central singularity is marginally naked ($\sfr{\alpha}{\beta}=1$),
\end{enumerate}
then the limit of $\sfr{|{\cal E}|}{\mu+p}$ at the central singularity along the apparent horizon curve diverges.
\end{prop}
\begin{proof}
We know that for a perfect fluid we have $Q=\Pi=0$, and at the central singularity $\sfr{\alpha}{\beta}=1$ implies
\be
\frac{\sfr23\mu+{\cal E}}{-\sfr13\mu-p+{\cal E}}=1\;,
\ee 
which can be simplified to
\be
\left[\frac{{\cal E}}{\mu+p}-\frac13\frac{\mu+3p}{\mu+p}\right]^{-1}=0.
\ee
For the perfect fluid satisfying the strong energy condition, $\sfr{\mu+3p}{\mu+p}$ is finite and hence $\sfr{|{\cal E}|}{\mu+p}$ at the central singularity 
along the apparent horizon tends to infinity.
\end{proof}
The above result clearly shows that the electric part of the Weyl scalar (which is responsible for the tidal forces) must diverge faster than the energy density along the 
apparent horizon curve, for a singularity to be locally naked. In fact, this results closely relates to the result obtained in \cite{Joshi:2001xi}. Equation (17) of that paper shows 
that the square of the shear scalar `$\sigma^2$' must diverge faster than the energy density at the central singularity of the collapsing dust.
\begin{cor}
Consider the continued gravitational collapse of a spherically symmetric perfect fluid obeying the strong energy condition $\mu\ge0$ and $\mu+3p\ge0$. If the 
spacetime is conformally flat then the end state of the collapse
is necessarily a black hole.
\end{cor}
\begin{proof}
Conformally flat spacetime implies vanishing of the Weyl tensor. Hence we have ${\cal E}=0$. Also for a perfect fluid $Q=\Pi=0$. We therefore have 
\be 
\frac{\alpha}{\beta}=\frac{-\sfr23\mu}{\sfr13\mu+p}\;.
\ee
Now the condition $\sfr{\alpha}{\beta}\ge1$ implies $\mu+p\le0$ which violates the strong energy condition.
In fact one can explicitly calculate the norm of the tangent to show that
\be
\Psi^a\Psi_a \propto -\frac{1}{3}(\mu+p)(\mu-3p)\;.
\ee
If the strong energy condition is satisfied  we have $\mu+p>0$, then we have the following cases:
\begin{enumerate}
\item If $\mu>3p$ the $\Psi^a$ is ``ingoing'' timelike.
\item If $\mu=3p$ the $\Psi^a$ is ``ingoing'' null.
\item If $\mu<3p$ the $\Psi^a$ is ``ingoing'' spacelike.
\end{enumerate}
In all these cases the region around the centre gets trapped before the central singularity. Hence the singularity is always covered and the collapse end-state 
is always a black hole.
\end{proof}
The above proposition highlights the importance of tidal forces in delaying the trapping. Absence of the Weyl tensor necessarily implies the absence of any tidal stresses, 
and we can easily see that the trapping occurs before the singularity formation.
\section{Some specific examples}
In this section we briefly discuss some of the well known examples of gravitational collapse scenarios in the light of the discussion in previous sections. As we will see below, 
in all these cases we can transparently determine the end state of the continued gravitational collapse using the formalism developed in this paper.
\subsection{Oppenheimer-Snyder dust collapse}
This was the first theoretical model of continued gravitational collapse, where the collapsing matter was assumed to be dustlike and homogeneous. In this case the interior 
metric is the Friedmann-Lemaitre-Robertson-Walker (FLRW) spacetime and is given by 
\be
ds^2=-dt^2+\frac{a(t)^2}{1-kr^2}dr^2+r^2a(t)^2(d\Theta^2+\sin^2\Theta d\phi^2)\;.
\ee
The FLRW metric is conformally flat and hence ${\cal E}=0$. Moreover, since the matter is dustlike we have $p=0$. Hence the slope of the tangent to the central singularity, 
$\sfr{\alpha}{\beta}=-2$. Thus the apparent horizon is ingoing timelike and the end state of the collapse is a black hole.
\subsection{Lemaitre-Tolman-Bondi dust collapse}
This is a well known gravitational collapse model where the Cosmic Censorship Conjecture is violated. Ever though the collapsing matter is dustlike it may be inhomogeneous. The interior of the collapsing dust is described by the LTB metric
\be
ds^2=-dt^2+\frac{R'^2}{1-r^2b_0(r)}dr^2+R^2(d\Theta^2+\sin^2\Theta d\phi^2)\;.
\ee
Here $R(t,r)$ is the area radius of the collapsing dust shell and $b_0(r)$ denotes their energy profile. The system is specified by two free functions at the initial epoch, 
the energy profile $b_0(r)$ and the initial mass profile $F(r)\equiv r^3{\cal M}(r)$. From the Einstein field equations we have 
\be\label{fe}
F'=\mu R^2R'.
\ee
If we consider the marginally bound case where $b_0(r)=0$, then the equation of motion of the 
collapsing shells are given by \cite{Goswami:2004gy, Goswami:2002he}
\be\label{R}
\dot{R}^2=\sfr{F}{R},
\ee
and the electric part of the Weyl scalar is \cite{Zibin:2008vj}
\be\label{calE}
 {\cal E}= \frac{1}{3} \mu-\frac{r^3{\cal M}(r)}{R^3}\;.
\ee
Following  \cite{Goswami:2004gy, Goswami:2002he}, we can write $R=ra(r,t)$ where $a(r,t)$ is the `scale factor' for a shell labelled `$r$'. Also we consider 
a smooth density profile at the centre and hence write the function ${\cal M}(r)\equiv{\cal M}_0+{\cal M}_2r^2$. We know that for the singularity curve to 
be an increasing function of `$r$' to avoid shell crossings etc, we must have ${\cal M}_2<0$. Solving the equation 
of motion we get 
\be\label{a}
a(r,t)=\bra{1-\sqrt{{\cal M}(r)}t}^{2/3}.
\ee
We can easily check (from Proposition 2 and Einstein's equations) that the equation of the apparent horizon is given by  
$F=R$. Now the slope of the horizon is given by
 \be\label{alph}
\frac{\alpha}{\beta}=\frac{\sfr23\mu+{\cal E}}{-\sfr13\mu+{\cal E}}\;.
\ee
Calculating the slope at the central singularity (given by $t=t_{s_0}=1/\sqrt{{\cal M}_0}$ and $r=0$) and using (\ref{fe}, \ref{a}) we get 
\be
\frac{\alpha}{\beta}=\lim_{t\rightarrow t_{s_0}}\lim_{r\rightarrow 0} 1-\frac{F'}{R'}=1.
\ee
Hence we see that provided ${\cal M}_2<0$, the central singularity will be locally naked.

\section{Discussion}

In this paper, working in a covariant  and frame independent formalism, we successfully identified the physical and geometrical mechanisms responsible 
for delaying the trapped surface formation and making the central singularity locally naked during the continued gravitational collapse of a massive star. 
By working out the dynamics of the trapped region we transparently and quantitatively identified the role of Weyl curvature in deforming the trapped region in such a 
way that the singularity can be naked. As we know the Weyl curvature is responsible for the tidal force between nearby geodesics that generates the 
spacetime shear. In fact from the field equations (\ref{Sigthetadot} and \ref{Raychaudhuri}) for LRS-II spacetimes one can immediately see that the Weyl scalar is the source term for the 
shear evolution equation. Spacetime shear then deforms the apparent horizon and delays the trapping as shown in  \cite{Joshi:2001xi, Joshi:2004tb}.

These findings can have possible important observational signatures that can identify black holes from a naked singularity, and hence observationally test 
the weak censorship hypothesis \cite{Kong:2013daa}. As we have seen, the Weyl curvature is the key feature that can generate a locally visible 
singularity. Moreover Weyl curvature is also the generator of gravitational waves \cite{Clarkson:2007yp}. Hence one can expect signatures of 
locally naked singularities from the gravitational waves radiated from a collapsing star.

\begin{acknowledgments}
AH would like to thank Radouane Gannouji for the useful discussions. AH and RG are supported by National Research Foundation (NRF), South Africa. SDM 
acknowledges that this work is based on research supported by the South African Research Chair Initiative of the Department of
Science and Technology and the National Research Foundation.
\end{acknowledgments}
%\appendix
%\section{First order}

\end{document}